\documentclass[12pt,leqno]{article}
\usepackage{oxford3}
\usepackage{graphicx}
\usepackage{amsfonts}
\usepackage{amsmath}
\usepackage{rotating}
\usepackage{enumerate}
\usepackage{amssymb, latexsym,amsthm}
\usepackage{multirow,array,epsfig}
\usepackage{verbatim}
\usepackage{color}
\usepackage{mathtools}

\theoremstyle{definition}
\newtheorem{algorithm}{Algorithm}[section]
\newtheorem{assumption}{Assumption}[section]

\newtheorem{procedure}{Procedure}[section]

\theoremstyle{theorem}
\newtheorem{theorem}{Theorem}[section]

\newtheorem{proposition}{Proposition}[section]

\numberwithin{equation}{section}

 \newcommand \bg{\beta}
 
 \newcommand \Gg{\Gamma}

 \newcommand \eg{\varepsilon}

 \newcommand \sg{\sigma}
 
\newcommand\bGa{\mbox{\boldmath${\Gamma}$}}
\newcommand\bGg{\mbox{\boldmath${\Gg}$}}

\newcommand\bSig{{\boldsymbol \Sigma}}
\newcommand\bXi{\mbox{\boldmath${\Xi}$}}

\newcommand\bone{\mbox{\boldmath${1}$}}

\newcommand\bA{{\bf A}}

\newcommand\bD{{\bf D}}
\newcommand\bE{{\bf E}}
\newcommand\be{{\bf e}}
\newcommand\bI{{\bf I}}

\newcommand\bG{{\bf G}}

\newcommand\bM{{\bf M}}

\newcommand\bQ{{\bf Q}}

\newcommand\bs{{\bf s}}

\newcommand\bU{{\bf U}}

\newcommand\bV{{\bf V}}

\newcommand\bW{{\bf W}}

\newcommand\bX{{\bf X}}

\newcommand\bZ{{\bf Z}}

\newcommand\cD{{\mathcal D}}

\newcommand\cI{{\mathcal I}}

\newcommand\cS{{\mathcal S}}
\newcommand\cT{{\mathcal T}}
\newcommand\cU{{\mathcal U}}
\newcommand\cV{{\mathcal V}}

\newcommand\tp{{t^\prime}}

\DeclareMathOperator{\E}{E}
\DeclareMathOperator{\vc}{vec}
\DeclareMathOperator{\Cov}{Cov}
\DeclareMathOperator{\tr}{tr}
\DeclareMathOperator{\Var}{Var}

\newcommand{\aln}[1]{\begin{align*}#1\end{align*}}

\begin{document}

\title{Testing separability of space--time  functional processes}
\author{Panayiotis Constantinou\\
{\small Pennsylvania State University}
\and
Piotr Kokoszka\\
{\small Colorado State University}
\and
Matthew Reimherr\thanks{\small {\em Corresponding author:}
Department of Statistics,
Pennsylvania State University, 411 Thomas Building,
University Park, PA 16802, USA. \
mreimherr@psu.edu \
(814) 865-2544 } \\
{\small Pennsylvania State University}
}
\date{}
\maketitle


\begin{abstract}
We present a new methodology and accompanying theory to test for separability of
spatio--temporal functional data.  In spatio--temporal
statistics, separability is a common simplifying assumption concerning the covariance structure which, if
true, can greatly increase estimation accuracy and inferential power.
While our focus is on testing for the separation of space and time in
spatio-temporal data, our methods can be applied to any area where
separability is useful, including biomedical imaging.
We present three tests, one being a functional extension of
the Monte Carlo likelihood method of
\citetext{mitchell:genton:gumpertz:2005}, while the other two are
based on quadratic forms.  Our tests are based on asymptotic
distributions of maximum likelihood estimators, and do not require
Monte Carlo or bootstrap replications. The specification of the joint
asymptotic distribution of these estimators is the main theoretical
contribution of this paper. It can be used to derive many other tests.
The main methodological finding is that one of the quadratic form
methods, which we call a \textit{norm approach}, emerges as a clear
winner in terms of finite sample performance in nearly every setting
we considered.  The norm approach focuses directly on the Frobenius
distance between the spatio--temporal  covariance function and its separable
approximation.  We demonstrate the efficacy of our methods via
simulations and an application to Irish wind data.
\end{abstract}

\section{Introduction} \label{s:i}
The  assumption of separability is used heavily in spatio--temporal
statistics,
\citetext{haas:1995}, \citetext{genton:2007}, \citetext{hoff:2011},
\citetext{paul:peng:2011},  \citetext{sun:li:genton:2012},
among many others. It is introduced in many textbooks, e.g.
\citetext{schabenberger:gotway:2005},
\citetext{sherman:2011}.
 Separability means that the spatio--temporal covariance structure
 factors into the product of two functions, one depending only on
 space, the other only on time.  Such an assumption provides a number
 of benefits.  From a modeling perspective, it allows one to draw on
 the large literature on covariance structures for spatial or temporal
 data.  The simpler structure induced by separability is then much
 easier to estimate than a nonseparable structure.  In the context of
 multivariate spatio--temporal data, the separability assumption can
 be stated in terms of the factorization of the covariance matrix.
 For more complex spatio--temporal data structures, analogous
 definitions can be formulated, as we explain below.  The work
 presented in this paper is motivated by {\em geostatistical
 functional data}; functions are observed at a number of spatial
 locations, though our methods can be readily generalized to a number
 of areas. For example, in biomedical imaging,  such as fMRI, one often
 has data in both space (the brain) and time,  separability can greatly
 simplify modeling.  In our context, separability
 implies that the optimal functions used for (temporal) dimension
 reduction are the same at every spatial location; information can
 then be pooled across spatial locations to get very good estimates of
 these functions. 
 Geostatistical functional data are quite common.  
 Perhaps the best known example is provided by annual
 temperature and log--precipitation curves (averaged over several
 decades) at several dozen locations in Canada. These data have been
 used in many examples in the monograph of of
 \citetext{ramsay:silverman:2005} and many research papers that
 followed, \citetext{delicado:2010} provide further references. Our
 own work has been concerned with such data as well;
 \citetext{gromenko:kokoszka:2012big} and Gromenko and Kokoszka
 (\citeyear{gromenko:kokoszka:2012mf},
\citeyear{gromenko:kokoszka:2013np}) study curves describing the
evolution of certain ionospheric parameters measured at
globally distributed locations at which radar--type instruments called
ionosondes operate. In \citetext{gromenko:kokoszka:reimherr:2015},
we study precipitation measurements extending over several
decades at about sixty locations in the Midwest.

Tests of separability for spatio--temporal covariances of {\em scalar}
fields are reviewed in Mitchell {\it et al.}
(\citeyear{mitchell:genton:gumpertz:2005},
\citeyear{mitchell:genton:gumpertz:2006}) and \citetext{fuentes:2006}.
If the spatio--temporal covariance has a specific parametric form, a
likelihood ratio test is possible. A similar parametric approach, in
conjunction with bootstrap, is taken by \citetext{liu:ray:hooker:2014}
in the context of functional data.
\citetext{mitchell:genton:gumpertz:2006} introduce a more general,
nonparametric  LRT
test, which requires that the number of repeated measurements be
greater than the product of the number of spatial and temporal
locations. We explain their idea in greater detail in the following.
\citetext{mitchell:genton:gumpertz:2005} explain how to
deal with this restrictions by dividing the temporal domain into
blocks. The test of \citetext{fuentes:2006} is based on the spectral
representation which assumes that the data are available on a spatial
grid. The data that motivate this research are not of this form.

We now explain the contribution of this paper in more specific terms.
It is instructive to begin by summarizing the procedure of
\citetext{mitchell:genton:gumpertz:2006}.  Suppose we observe $N$ iid
scalar fields at temporal points $t_i$ and spatial locations $s_k$,
so that the data are replications of the spatio--temporal observations:
\[
X_n(\bs_k; t_i), \ \ \ 1 \le k \le K, \ \ 1 \le i \le I.
\]
The iid assumption implies  that the mean function
$
E  X_n(\bs; t) = \mu(\bs; t)
$
is the same for each replication. The covariances
\[
\sg_{k\ell; ij} =
E \left\{ \left [ X_n(\bs_k; t_i) - \mu(\bs_k; t_i) \right ]
\left [ X_n (\bs_\ell; t_j) - \mu(\bs_\ell; t_j) \right ]
\right\}
\]
do not depend on $n$ either. The assumption of separability implies
that $ \sg_{k\ell; ij} = u_{k\ell} v_{ij}, $ where $u_{kl}$ does not
depend on $t$ (time) and $v_{ij}$ does not depend on $\bs$ (space).
This relation stated in the matrix form as
\begin{equation} \label{e:sep-mat}
\bSig = \bV \otimes \bU
=
\begin{bmatrix}
v_{11} \bU & v_{12} \bU & \ldots & v_{1I} \bU\\
v_{21} \bU & v_{22} \bU & \ldots & v_{2I} \bU\\
\vdots         &  \vdots        & \vdots & \vdots   \\
v_{I1} \bU & v_{I2} \bU & \ldots & v_{II} \bU\\
\end{bmatrix},
\end{equation}
where $\bU$ is the $K\times K$ matrix with entries $u_{k\ell}$ and
$\bV$ is the $I\times I$ matrix with entries $v_{ij}$.  The matrix
$\bSig$ is $KI \times KI$ and can be viewed as the covariance matrix
of the vectorized matrix $\{X_n(s_k, t_i)\}$ (with $k$ indexing rows
and $i$ columns).  \citetext{mitchell:genton:gumpertz:2006} use the
test statistic
\begin{equation} \label{e:mgg}
\widehat T = N \left\{ K \log \det[\widehat\bV] + I \log \det[\widehat\bU]
- \log \det[\widehat\bSig]\right\},
\end{equation}
where $\widehat\bV, \widehat\bU, \widehat\bSig$ are Gaussian
likelihood estimates defined in Theorem~\ref{t:mult}.  Their approach
is based on Theorem \ref{t:mgg} which justifies a Monte Carlo
approximation for the null distribution of $\widehat T$.  One can,
e.g., use $\mu=0, \bU= \bI_{K}, \bV= \bI_{I}$ to obtain a large number
of replicates of $\widehat T$, and so approximate its null distribution.

\begin{theorem}[\citetext{mitchell:genton:gumpertz:2006}] \label{t:mgg}
If the observations are normal, \eqref{e:sep-mat}
holds, and
$
N > KI,
$
then the distribution of  $\widehat T$  defined
in \eqref{e:mgg} does not depend on $\mu, \bU, \bV$.
\end{theorem}

The choice of statistic \eqref{e:mgg} is thus fundamentally justified
by the invariance property stated in Theorem~\ref{t:mgg}.
Perhaps more natural test statistics should be based on some
distance between the matrices $\widehat \bSig$ and
$\widehat\bV \otimes \widehat\bU$. It might be hoped that
a more direct comparison would lead to tests with better
power. However, such statistics are not invariant
in the sense of Theorem~\ref{t:mgg} and their
asymptotic distribution has not been found. The first
contribution of this paper is to derive  the joint
asymptotic null distribution of $\widehat \bSig,
\widehat\bV,  \widehat\bU$ and show how it enables
to derive the limit distribution a several natural test statistic in
the multivariate context. This is addressed in
Section~\ref{s:mult-th}. The proofs are presented in the Appendix.

The second contribution, which motivated our research, is
related to functional data which  are $N$  replications
of the field
\[
X_n(\bs_k,t), \ \ 1 \leq k \leq K, \ \  t \in \cT.
\]
At each spatial location $\bs_k$,  a function with argument $t$ is
observed.  For example, $X_n(\bs_k, t)$ can be the maximum daily
temperature on day $t$ of year $n$ at location $\bs_k$. For historical
climate and environmental data sets of this type, $N$ is about 100,
$K$ can be anything from several dozen to a few hundred, and the
number of measurements per year is $I=365$.  The approach of
\citetext{mitchell:genton:gumpertz:2006} thus cannot be applied
because the condition $N> KI$ is violated. One cannot subdivide the
years into smaller units because the assumption of the identical
distribution would be violated, so the modification of
\citetext{mitchell:genton:gumpertz:2005} cannot be applied.
Our approach exploits the functional structure of the data and uses a
dimension reduction.  Since for historical climate and environmental
data sets $N$ is fairly large, we derive asymptotic tests as $N\to
\infty$ using the results of  Section~\ref{s:mult-th}.  These
developments are described in Section~\ref{s:test}.

The only related work we are aware of which also concerns functional data is given in the (currently) unpublished work of 
\citetext{aston:pigoli:tavakoli:2015}.  Since both
papers, developed independently and concurrently, are currently unpublished, it is
only appropriate to discuss differences and similarities, refraining
from any evaluative statements. Both papers aim at solving the
same problem, with motivation coming from different data structures.
Our work is motivated by geostatistical functional data, curves
observed at irregularly distributed spatial locations;
\citetext{aston:pigoli:tavakoli:2015} are motivated by data observed
on grids with one dimension which can be called space and the other
time. In their application to phonetic data, space is frequency. Our
approach is based on the joint asymptotic distribution of the
MLE's,  followed, as an option, by dimension reduction;
\citetext{aston:pigoli:tavakoli:2015} first perform dimension
reduction in space and time which allows them to compute
their tests statistic without estimating the full spatio--temporal
covariance. Our method uses an approximation via limiting
distributions; they use bootstrap approximations. The computational
efficiency is thus obtained in very different ways. Regarding asymptotic
theory, ours is based on joint MLE's which require an iterative
procedure to compute;  \citetext{aston:pigoli:tavakoli:2015}
compute the marginal space and time covariances by
integrating out the other dimension,  and then apply the CLT
to the difference projected on a finite number of tensor products.

The remainder of this paper is organized as follows.
Section~\ref{s:sim} compares the finite sample performance of the
tests  derived Section~\ref{s:test}. We show that a test
related to a norm of the difference $\widehat \bSig - \widehat\bV
\otimes \widehat\bU$ has correct size and is most powerful.
It is also more powerful than a Monte Carlo test based on
Theorem~\ref{t:mgg} (with a suitable extension to functional data).
 We apply the tests to an extensively studied Irish
wind data set,  and confirm the conjecture of T. Gneiting that these
space--time data are not separable.

\section{Multivariate theory} \label{s:mult-th}
This section clarifies the behavior of several test statistics based
on normal maximum likelihood estimators. The theory is valid under
the following assumption.
\begin{assumption} \label{a:mult}
Assume $\bX_1, \ldots, \bX_N$ are iid normally
distributed $K \times I $ matrices  with $\E[\bX_n] =\bM$
and $\Cov(\vc(\bX_n)) = \bSig$, where $\bSig$ is
an $KI\times KI$ matrix of full rank.
\end{assumption}

We begin with Theorem~\ref{t:mult} whose proof is based on direct, but
lengthy and tedious calculations of Gaussian likelihoods for
vectorized matrices of various dimensions and solving score equations.
It is placed in the supplement.   Recall that if
$\bA$ is a $K \times I$ matrix, then $\vc(\bA)$ is a column vector of
length $KI$ obtained by stacking the columns of $\bA$ on top of each
other.

\begin{theorem} \label{t:mult}
Under Assumption~\ref{a:mult},  the maximum
likelihood estimators of $\bM$ and $\bSig$ are
\[
\widehat\bM = \frac{1}{N} \sum_{n=1}^N \bX_n, \ \ \ \
\widehat\bSig = \frac{1}{N} \sum_{n=1}^N
\vc(\bX_n - \widehat\bM) (\vc(\bX_n - \widehat\bM))^\top.
\]
If $\bSig$ admits the decomposition
\begin{equation} \label{e:bSigDec}
\bSig = \bV \otimes \bU, \ \ \ {\rm dim} (\bU) = K \times K , \ \
{\rm dim} (\bV) = I\times I,
\end{equation}
where $\bU$ and $\bV$ are of full rank,
then the maximum likelihood estimators of  $\bU$ and $\bV$ satisfy
\begin{align*}
\widehat \bU & = \frac{1}{NI} \sum_{n=1}^N
(\bX_n - \widehat\bM) {\widehat\bV}^{-1} (\bX_n - \widehat\bM)^\top \\
\widehat\bV & =  \frac{1}{NK} \sum_{n=1}^N
(\bX_n - \widehat\bM)^\top {\widehat\bU}^{-1} (\bX_n - \widehat\bM).
\end{align*}
\end{theorem}

The estimators $\widehat \bU$ and $\widehat\bV$ are defined
indirectly and must be  computed using an iterative
procedure with some normalization to ensure identifiability.
The following algorithm, \citetext{dutilleul:1999}, produces
$O_P(N^{-1/2})$ consistent estimators. It uses the normalization
 $\tr(\bU_i) = K$.

\begin{algorithm} \label{alg:duti}
Initialize with $\bU_0 = \bI_K$ ($K\times K$
 identity matrix).  For $i=0, 1, 2, \ldots$ calculate
\begin{align*}
\widehat\bV_i & =  \frac{1}{NK} \sum_{n=1}^N
(\bX_n - \widehat\bM)^\top {\widehat\bU}_i^{-1} (\bX_n - \widehat\bM),\\
\widetilde \bU_{i+1} & = \frac{1}{NI} \sum_{n=1}^N
(\bX_n - \widehat\bM) {\widehat\bV}_i^{-1} (\bX_n - \widehat\bM)^\top,  \\
\widehat \bU_{i+1}
& = \frac{K}{\tr(\widetilde \bU_{i+1})} \widetilde \bU_{i+1},
\end{align*}
until convergence is reached.
\end{algorithm}

The most natural statistic to test separability, i.e. $\bSig = \bV \otimes \bU$,
should be based on a difference between
$\widehat\bV \otimes \widehat\bU $ and $\widehat\bSig$. We will
show that the statistic
\[
\widehat T_{F}
= N \| \widehat\bV \otimes \widehat\bU  - \widehat\bSig \|_F^2,
\]
where $\| \cdot \|_F^2$ is the squared Frobenius matrix norm (i.e.
the sum of squares of all entries), converges, and find the asymptotic
distribution. This distribution involves the asymptotic covariance
matrix of $\vc(\widehat\bV \otimes\widehat\bU - \widehat\bSig )$,
which we denote by $\bW$. The form of $\bW$ is complex, see
\eqref{e:wmatrix}. To obtain a chi--square limit distribution, a suitable
quadratic form must be used. This leads to the statistic
\[
\widehat T_{W}
= N \vc(\widehat\bV \otimes \widehat\bU  - \widehat\bSig )^\top
 \widehat \bW^{+} \vc(\widehat\bV \otimes \widehat\bU  - \widehat\bSig ),
\]
where $\widehat \bW$ is an estimator of $\bW$ and $\widehat \bW^{+}$
is its generalized inverse.  A generalized inverse must be used
because $\bU$, $\bV$, and $\bSig$ (and the corresponding estimates)
are all symmetric, and this implies many linear constraints, $\hat
U_{k\ell} = \hat U_{\ell k}$ for example, on the entries of
$\vc(\widehat\bV \otimes\widehat\bU - \widehat\bSig )$ and so of
$\widehat
\bW$.  Using a generalized inverse is equivalent to dropping redundant
entries.

We will also show that the likelihood ratio statistic
\[
\widehat T_{L} =
N\left ( T \log \det (\widehat\bU)  + K \log\det(\widehat\bV)
- \log\det(\widehat\bSig) \right )
\]
discussed in Section~\ref{s:i} has the same limit as $\widehat T_{W}$,
i.e. is asymptotically chi--square with a known number of degrees of
freedom.  The asymptotic chi--square distribution of $\widehat T_{L}$
was claimed by \citetext{lu:zimmerman:2005} without proof. We present
a detailed proof in Section~\ref{ss:mult}.
\citetext{mitchell:genton:gumpertz:2006} did not use this asymptotic
result; they utilized a Monte Carlo finite sample approximation based
on Theorem~\ref{t:mgg}.  We collect our results in
Theorem~\ref{t:mult-asy}.

\begin{theorem} \label{t:mult-asy} Suppose Assumption~\ref{a:mult} and
  decomposition \eqref{e:bSigDec} hold. Let $\bW$ be the $KI \times
  KI$ matrix defined in \eqref{e:wmatrix}, with
  $\gamma_1,\gamma_2,\dots,\gamma_R$ its eigenvalues.  Then, as $N\to
  \infty$,
\[
\widehat T_{L} \sim \chi^2_{d} \quad \mbox{and} \quad \widehat T_{W} \sim \chi^2_{d}
\]
where
\[
d = \frac{KI(KI +1)}{2} -  \frac{K(K +1)}{2} -  \frac{I(I +1)}{2}+1,
\]
and
\[
\widehat T_{F} \sim \sum_{r=1}^R \gamma_r \chi^2_1(r),
\]
where $\{\chi^2_1(r)\}$ are iid chi-square random variables with one
degree of freedom.
\end{theorem}

Theorem~\ref{t:mult-asy} is proven in Section~\ref{ss:mult}. The
three statistics listed in Theorem~\ref{t:mult-asy} are not the
only ones that our theory covers. Theorem~\ref{t:joint},
which specifies the joint asymptotic distribution of
$\widehat\bV, \widehat\bU$ and $\widehat\bSig$,  can be used
to derive the asymptotic distribution of many other reasonable
test statistics.

\section{Tests for functional data} \label{s:test} 
We now show how the
results of Section~\ref{s:mult-th} are applied to testing separability
of geostatistical functional data.  For a reader interested in learning more about functional data methods there are now several introductory books including \citetext{ramsay:silverman:2005,ramsay:hooker:graves:2009,HKbook,hsing:eubank:2015}. 
We consider independent
spatio--temporal random fields $X_n(\cdot, \cdot), 1 \le n \le N,$
which have the same distribution as the field $\left \{ X(\bs, t), \bs
  \in \cS, t \in \cT\right \}$, which satisfies $ E \int_{\cS}
\int_\cT X^2(\bs, t ) d\bs dt < \infty.  $ Then $ X_n(\bs, t) =
\mu(\bs, t) + \eg_n(\bs, t), $ where $\mu \in L^2(\cS\times \cT)$, and
the $\eg_n$ are iid random elements of $ L^2(\cS\times \cT)$ which
satisfy
\[
E \int_{\cS} \int_\cT  \eg_n^2(\bs, t ) d\bs dt < \infty
\ \ \ {\rm and} \ \ \ E  \eg_n(\bs, t ) = 0.
\]
We consider the covariances
\[
\sg(\bs, \bs^\prime; t, t^\prime) = \Cov(X(\bs, t), X(\bs^\prime, t^\prime))
= E[\eg_n(\bs; t) \eg_n(\bs^\prime, t^\prime)].
\]
Our objective is to test
\begin{equation} \label{e:H0}
H_0: \ \ \  \sg(\bs, \bs^\prime; t, \tp) = \cU(\bs, \bs^\prime) \cV(t, \tp).
\end{equation}

As in the multivariate case, the functions $\cU$ and $\cV$ are
uniquely determined only up to multiplicative constants, and a testing
algorithm must include some arbitrary normalization. However, the
P--values of our testing procedures do not depend on this choice.  We
now proceed with the description of  test procedures of
increasing  complexity.

\subsection{Procedure 1: fixed spatial locations, fixed temporal
  basis} \label{ss:p1} We assume that for each $n$ the field $X_n$ is
observed at the same spatial locations $\bs_k, 1 \le k \le K$. We
estimate $\mu(\bs_k, t)$ by the sample average $ \hat\mu(\bs_k,t) =
N^{-1} \sum_{n=1}^N X_n(\bs_k,t), $ and focus in the following on the
covariance structure.
Under $H_0$, the covariances of the observations  are
\[
\Cov ( X_n(\bs_k, t) X_n(\bs_\ell, \tp)) =U(k, \ell) \cV(t,\tp),
\]
with entries $U(k,\ell) = \cU(\bs_k, \bs_\ell)$ forming a $K\times K$
matrix $\bU$, and $\mathcal V$ being the temporal covariance function
over $\cT \times \cT$.  As in the multivariate setting, the estimation
of the matrix $\bU$ and the covariance function $\cV$ must involve an
iterative procedure. In the functional setting, a dimension reduction
is also needed.  Suppose $\{v_j, j \ge 1 \}$ is a basis system in
$L^2(\cT)$ such that for sufficiently large $J$, the functions
\[
X_n^{(J)}(\bs_k, t) = \mu(\bs_k, t) + \sum_{j=1}^J \xi_{jn}(\bs_k) v_j(t)
\]
are  good approximations to the functions $X_n(\bs_k)$.
We thus replace a  large number of time points by a moderate number
$J$,  and seek to reduce the testing of $H_0$  \eqref{e:H0} to testing
the separability of the covariances of the transformed observations given as
$K\times J$ matrices
\begin{equation} \label{e:Xi-n}
\bXi_n = [\xi_{jn}(\bs_k),\  1 \le k \le K,\ 1 \le j \le J].
\end{equation}
The index $j$ should be viewed as a transformed time index.
The number $I$ of the actual time points $t_i$ can be very
large, $J$ is usually much smaller.
The following proposition is easy to prove. It establishes the
connection between the testing problem \eqref{e:H0} and
testing the separability of the transformed data \eqref{e:Xi-n}.
The assumption that the $v_j$ are orthonormal cannot be
removed.

\begin{proposition} \label{p:conect}
For some  orthonormal $v_j$, set
\[
X^{(J)}(\bs_k, t) = \mu(\bs_k, t) + \sum_{j=1}^J \xi_{j}(\bs_k) v_j(t).
\]
If
\begin{equation} \label{e:3.1}
E[\xi_j(\bs_k) \xi_i(\bs_\ell)] = U(k, \ell) V(j,i),
\end{equation}
then
\begin{equation}\label{e:3.2}
\Cov( X^{(J)}(\bs_k, t),  X^{(J)}(\bs_\ell, s) ) = U(k, \ell) \cV(t,s).
\end{equation}
Conversely,  \eqref{e:3.2} implies \eqref{e:3.1}.
The entries  $V(j,i)$ and  $\cV(t,s)$ are related via
\[
V(j,i) = \iint \cV(t,s) v_j(t) v_i(s) dt ds, \ \ \
\cV(t,s) = \sum_{j,i=1}^J V(j,i)  v_j(t) v_i(s).
\]
\end{proposition}

We assume that  $\left\{ v_j, \ 1 \le j \le J\right\}$ is a fixed
orthonormal system, for example the first $J$ trigonometric basis
functions. Slightly abusing notation, consider the matrices
\[
\widehat\bSig \ (KJ\times KJ), \ \ \
\widehat\bU \ (K\times K), \ \ \ \widehat\bV \ (J\times J),
\]
defined as in Theorem~\ref{t:mult}, but with matrices $\bX_n$
replaced by the matrices $\bXi_n$. The index $ j \le J$ now plays
the role of the index $i\le I$ of Section~\ref{s:mult-th}.
To apply tests based on Theorem~\ref{t:mult-asy},
we must recursively calculate $\widehat\bU$ and $\widehat\bV$
using the relations stated in Theorem~\ref{t:mult}. This
can be done using Algorithm~\ref{alg:duti} with
$\bXi_n$ in place of $\bX_n$.
This approach leads to the following test procedure. The
test statistic can be one of the three statistics
introduced in Section~\ref{s:mult-th}.

\begin{procedure} \label{proc:no-PC} \hspace*{15cm}

\noindent {\bf 1.} Choose a deterministic orthonormal basis $v_j, j \ge 1$.

\noindent {\bf 2.} Approximate each curve  $X_n(\bs_k, t)$ by
\[
X_n^{(J)}(\bs_k, t) = \hat\mu(\bs_k, t) + \sum_{j=1}^J
\xi_{jn}(\bs_k) v_j(t).
\]
Construct $K\times J$ matrices $\bXi_n$ defined in \eqref{e:Xi-n}.

\noindent {\bf 3.} Compute the matrix
\[
\widehat\bSig = \frac{1}{N} \sum_{n=1}^N
\vc(\bXi_n - \widehat\bM) (\vc(\bXi_n - \widehat\bM))^\top, \ \ \ \
\widehat\bM = \frac{1}{N} \sum_{n=1}^N \bXi_n.
\]
Using Algorithm~\ref{alg:duti} with
$\bXi_n$ in place of $\bX_n$,
compute
the matrices $\widehat\bU$ and  $\widehat\bV$.

\noindent {\bf 4.}
Estimate the matrix $\bW$ defined by \eqref{e:wmatrix} by
replacing $\bSig, \bU, \bV$ by their estimates.

\noindent {\bf 5.} Calculate the P--value using the limit distribution
specified in Theorem~\ref{t:mult-asy},  with $I$ replaced by $J$.

\end{procedure}

Step 2 can be easily implemented using  {\tt R} function {\tt pca.fd},
see \citetext{ramsay:hooker:graves:2009}.
Several methods of choosing $J$ are available; we  used the
cumulative variance rule requiring that $J$ be so large that
at least 80\% of variance is explained  for each location $\bs_k$.

\subsection{Procedure 2: fixed spatial locations, data driven temporal
basis} \label{ss:p2}
In Section~\ref{ss:p1}, we used a  {\em deterministic}
orthonormal system. To achieve the most
efficient dimension reduction, it is usual to project on a data driven
system, with the functional principal components being used most
often. Since the  sequences of functions are defined at a number of spatial
locations, it is not a priori clear how a suitable orthonormal system
should be constructed, as each sequence $\left\{ X_n(\bs_k), 1 \le n
\le N \right\}$ has different functional principal components
$v_j(\bs_k), j \ge 1$, and Proposition~\ref{p:conect} requires that a
single system be used.  Our next algorithm proposes an approach which
leads to suitable estimates $\widehat\bU,
\widehat\bV$ and $\widehat\bSig$. It is not difficult to
show that these estimators are $O_P(N^{-1/2})$ consistent.

\begin{algorithm} \label{alg:with-PC}
Initialize with $\bU_0 = \bI_K$.

For $i=1,2,\dots$, perform the following two steps
until convergence is reached.

\noindent{\bf 1.} Calculate
\[
\cV_i(t,\tp) = (NK)^{-1} \sum_{n=1}^N
(X_n(\cdot, t) - \hat\mu(\cdot, t))^\top \bU_{i-1}^{-1}
(X_n(\cdot, \tp) - \hat\mu(\cdot, \tp)).
\]
Denote the eigenfunctions and eigenvalues of $\cV_i$ by $\{v_{ij}\}$
and $\{\lambda_{ij}\}$. Determine $J_i$ such that the first $J_i$
eigenfunctions of $\cV_i$ explain at least 85\% of the variance.

\noindent{\bf 2.}
Project each function $X_n(\bs_k, \cdot)$ on the first $J_i$
eigenfunctions of $\mathcal V_i$.  Denote the scores of these  projections by
\[
Z_{in}(\bs_k,j) = \langle X_n(\bs_k, \cdot) - \hat\mu(\bs_k, \cdot),
v_{ij} \rangle
\]
and  calculate
\[
U_{i}(k , \ell) = (NJ_i)^{-1} \sum_{j=1}^{J_i}
 \sum_{n=1}^N \frac{Z_{in}(\bs_k,j) Z_{in}(\bs_\ell,j)}{\lambda_{ij}}.
\]
Normalize $\bU_i$ so that $\tr(\bU_i) = K$.

Let $\{\hat v_j, 1 \le j \le J\}$ denote the final eigenfunctions.
Carry out the final projection $\hat Z_n(\bs_k,j) = \langle X_n(\bs_k,
\cdot) - \hat\mu(\bs_k, \cdot), \hat v_j \rangle$.  For each $n$,
denote by $\widehat\bZ_n$ the $K \times J$ matrix with these entries.
Set
 \begin{equation} \label{e:W2}
\widehat \bSig = \frac{1}{N} \sum_{n=1}^N \vc(\widehat\bZ_n)
\vc(\widehat\bZ_n)^\top
\end{equation}
and apply Algorithm~\ref{alg:duti} with  $\bX_n=\widehat\bZ_n$ to
obtain $\widehat \bU$ and $\widehat \bV$.
\end{algorithm}

Using the above algorithm, the testing procedure is as follows:

\begin{procedure} \label{proc:with-PC} \hspace*{15cm}

\noindent {\bf 1.} Calculate matrices $\widehat\bSig, \widehat\bU,
  \widehat\bV$ according to Algorithm~\ref{alg:with-PC}.

\noindent {\bf 2.} Perform steps 4 and 5 of Procedure~\ref{proc:no-PC}.

\end{procedure}

\subsection{Procedure 3: dimension reduction in both space and time}
\label{ss:p3}

Similar with Procedure 1 we assume that for each $n$ the field $X_n$ is
observed at the same spatial locations $\bs_k, 1 \le k \le K$. We
estimate $\mu(\bs_k, t)$ by the sample average $ \hat\mu(\bs_k,t) =
N^{-1} \sum_{n=1}^N X_n(\bs_k,t), $ and focus in the following on the
covariance structure.
Under $H_0$, the covariances of the observations  are
\[
\Cov ( X_n(\bs_k, t) X_n(\bs_\ell, \tp)) =U(k, \ell) \cV(t,\tp),
\]
with entries $U(k,\ell) = \cU(\bs_k, \bs_\ell)$ forming a $K\times K$
matrix $\bU$, and $\mathcal V$ being the temporal covariance function
over $\cT \times \cT$.

\noindent The testing procedure for dimension reduction in both space and time is as follows:

\begin{procedure} \label{proc:both-PC} \hspace*{15cm}

\noindent {\bf 1.} Choose a deterministic orthonormal basis $v_j, j \ge 1$.

\noindent {\bf 2.} Approximate each curve  $X_n(\bs_k, t)$ by
\[
X_n^{(J)}(\bs_k, t) = \hat\mu(\bs_k, t) + \sum_{j=1}^J
\xi_{jn}(\bs_k) v_j(t).
\]
Construct $K\times J$ matrices $\bXi_n$ defined in \eqref{e:Xi-n} where $J$ is chosen so large that for each $k$ the first $J$ sample eigenvalues explain at least $80\%$ of the variance.  This is Functional Principal Components Analysis carried out on the pooled (across space) sample.

\noindent {\bf 3.} Approximate each vector $(\xi_{jn}(\bs_1), \dots, \xi_{jn}(\bs_K))$ using
\[
\xi_{jn}(\bs_k) = \sum_{l=1}^L \zeta_{lj;n} u_{l}(s_k).
\]
The vectors $(u_l(s_1),\dots, u_l(s_K))$ are the eigenvectors of the following matrix
\[
\tilde U_{i}(k , \ell) = (NJ_i)^{-1} \sum_{j=1}^{J_i}
 \sum_{n=1}^N \frac{\xi_{jn}(\bs_k) \xi_{jn}(\bs_\ell)}{\lambda_{ij}}.
\]
Construct the $L\times J$ matrices $\bZ_n = [\zeta_{lj;n},\  1 \le l \le L,\ 1 \le j \le J]$ where $L$ is chosen large enough so that the first $L$ eigenvalues explain at least $80\%$ of the variance.  This is a multivariate PCA on the pooled (across time) variance adjusted sample.

\noindent {\bf 4.} Compute the matrix
\[
\widehat\bSig = \frac{1}{N} \sum_{n=1}^N
\vc(\bZ_n - \widehat\bM) (\vc(\bZ_n - \widehat\bM))^\top, \ \ \ \
\widehat\bM = \frac{1}{N} \sum_{n=1}^N \bZ_n.
\]
Using Algorithm~\ref{alg:duti} with
$\bZ_n$ in place of $\bX_n$, $J$ in place of $I$ and $L$ in place of $K$
compute
the matrices $\widehat\bU$ and  $\widehat\bV$.

\noindent {\bf 5.}
Estimate the matrix $\bW$ defined by \eqref{e:wmatrix} by
replacing $\bSig, \bU, \bV$ by their estimates.

\noindent {\bf 6.} Calculate the P--value using the limit distribution
specified in Theorem~\ref{t:mult-asy},  with $I$ replaced by $J$ and $K$ replaced by $L$.

\end{procedure}

Step 2 can be easily implemented using  {\tt R} function {\tt pca.fd} and step 3 by using {\tt R} function {\tt prcomp}.

\section{Finite sample comparison and
application to Irish wind data} \label{s:sim}
We now compare the performance of the tests based on statistics
introduced in  Section~\ref{s:mult-th} and procedures introduced in
Section~\ref{s:test}.  We include in the comparison the modified
approach of  \citetext{mitchell:genton:gumpertz:2006} which
is based on Theorem~\ref{t:mgg} and the spatial principal
components introduced in Section~\ref{ss:p3}. We tabulate the results
for the most general approach described in Section~\ref{ss:p3},
which also leads to most accurate tests  for the simulated data
we used. The relative ranking of the tests remains
the same if the approaches  of Sections~\ref{ss:p1} and
\ref{ss:p2} are applied;  these approaches
perform best  if the number of spatial locations is small.
The approach based on  Theorem~\ref{t:mgg} can typically  be applied
only in conjunction with the procedure of Section~\ref{ss:p3}
so that the condition $N> LJ$ holds.

We Thus consider four test procedures applicable to space--time
functional data, which we denote $T_L, T_F, T_W$ and
$T_{L-MC}$. The first three procedures use
asymptotic critical values of limit distributions
specified in Theorem~\ref{t:mult-asy};
$T_{L-MC}$ uses the Monte Carlo critical values computed
using   $\mu=0, \bU= \bI_{L}, \bV= \bI_{J}$.

To generate data, we use the
following spatio--temporal  covariance function introduced by
\citetext{gneiting:2002}:
\begin{equation} \label{e:gneiting-cov}
\sigma(\bs, \bs^\prime,t, \tp)
=\frac{\sigma^2}{(a|t-t^\prime|^{2\alpha}+1)^\tau}\exp\left(-\frac{c\|\bs - \bs^\prime\|^{2\gamma}}{(a|t - \tp|^{2\alpha}+1)^{\beta\gamma}}\right)
\end{equation}
In this covariance function, $a$ and $c$ are nonnegative scaling
parameters of time and space respectively, $\alpha $ and $\gamma$ are
smoothness parameters which take values in $(0,1]$, $\beta$ is the
separability parameter which takes values in $[0,1]$, $\sigma^2>0$ is
the point-wise variance and finally $\tau\geq\beta d/2$, where $d$ is
the spatial dimension. We focus on the effect of the space--time
interaction parameter, $\beta\in[0,1]$. If $\beta=0$, the
covariance function is separable. As $\beta$ increases the space-time
interaction becomes stronger. We
set  $\gamma=1$, $\alpha=1/2$, $\sigma^2=1$, $a=1$, $c=1$ and
$\tau=1$
so the covariance function becomes:
\begin{equation*}
\sigma(\bs, \bs^\prime,t, \tp)=\frac{1}{(|t-\tp|+1)}
\exp\left(-\frac{\| \bs - \bs^\prime\|^{2}}{(|t -\tp|+1)^{\beta}}\right).
\end{equation*}
We use $I=100$ time points equally
spaced on $[0,1]$ and $K=11$ space points on a  grid in
$[0,1] \times [0,1]$.  The number $K=11$ is motivated by the Irish
wind data considered by \citetext{gneiting:2002}, which
we also study below in this section.   We will
consider different values of the parameter $\beta$ as well as the
number of spatial PC's, $L$, and temporal FPC's, $J$.  We will also
consider different values for the sample size $N$. All empirical
rejection rates are based on one thousand replications, so their
precision is about 0.7 percent for size  (we use significance
level of 5\%), and about two percent for
power.

We study three different scenarios. The first scenario considers
different values of $\beta$. The second scenario examines the effect of
the sample size $N$, while the third scenario the effect of the
number of principal components.  Each table reports  the
rejection rates in percent.

\noindent $\bullet$ Scenario 1: $N=100$, $L=J=2$, $\beta= 0, 0.5, 1$.

\begin{center}
\begin{tabular}{|l|r|r|r|r|}
\hline
& $T_{L-MC}$& $T_L$ & $T_F$ & $T_W$\\
\hline
$\beta=0$&5.1 &5.9 & 4.5 & 3.7\\
\hline
$\beta=0.5$&12.3 & 13.4 & 15.2& 5.4\\
\hline
$\beta=1$& 54.1 & 55.8 & 63.4& 32.3\\
\hline
\end{tabular}
\end{center}
The test $T_F$ has the best balance of size and power.
The test  $T_W$ is a bit conservative here, however
we will see in Scenario 3 that this pattern is not consistent.  The
two likelihood methods do not exhibit  significantly different
rejection rates.

\noindent $\bullet$  Scenario 2: $N=100, 150, 200$, $L=J=2$, $\beta=1$.

\begin{center}
\begin{tabular}{|c|c|c|c|c|}
\hline
& $T_{L-MC}$&$T_L$ & $T_F$ & $T_W$\\
\hline
$N=100$& 54.1& 55.8 & 63.4& 32.3\\
\hline
$N=150$&68.0 &68.3 & 79.8& 29.8\\
\hline
$N=200$&80.4 &80.8 & 91.5& 52.0\\
\hline
\end{tabular}
\end{center}As the sample
size increases, the empirical power is also increasing,
 with the $T_F$--test  preserving its lead in terms of power.

\noindent $\bullet$ Scenario 3: $N=100$, $\bg=0, 1 $, $L+J$ increasing.

\begin{center}
\begin{tabular}{|c|r|r|r|r|}
\hline
$\beta=0$&$T_{L-MC}$& $T_L$ & $T_F$ & $T_W$\\
\hline
$L=J=2$& 5.1 &5.9 & 4.5& 3.7\\
\hline
$L=2,J=3$&5.5 & 6.5 & 5.3& 29.5\\
\hline
$L=3,J=2$&4.6 &5.4 & 4.5& 10.2\\
\hline
$L=J=3$&4.4 &6.9 & 4.7 & 39.4\\
\hline
$L=J=4$& 5.2 &13.6 & 5.2 & 98.1\\
\hline
\hline
$\beta=1$ & & & & \\
\hline
$L=J=2$& 54.1 &55.8 & 63.4& 32.3\\
\hline
$L=2,J=3$& 52.3& 55.0 & 75.5& 91.6\\
\hline
$L=3,J=2$ & 47.3 &50.0 & 74.6& 59.9\\
\hline
$L=J=3$& 54.7 &61.9 & 89.1 & 95.8\\
\hline
$L=J=4$ &79.0 &90.5 & 99.7 & 100.0\\
\hline
\end{tabular}
\end{center}

Only the tests $T_{L-MC}$ and $T_F$ are robust to the number
of the principal components used. This a a very desirable property,
as in all procures of FDA there is some uncertainty as the actual
number of PC's that should be used. The test $T_F$ is more
powerful than  $T_{L-MC}$.

Our overall conclusion is that the norm based test, $T_F$, works better
than the other approaches.  This is due to the fact that it targets
the difference between $\bSig$ and $\bU\otimes \bV$ most directly.
The application of this norm based approach  is possible because
we have derived the asymptotic distribution of $\widehat T_F$.
Even though this distribution is very complex (the matrix $\bW$
defined in \eqref{e:wmatrix} has a complex structure), once
the algorithm is coded, the test can be applied without difficulty.

\medskip

We conclude this section by
considering  the Irish wind data of \citetext{haslett:raftery:1989}
which consists of daily averages of wind speeds at $11$ synoptic
meteorological stations in Ireland during the period $1961-1978$. The
data are available at Statlib, \\ http://lib.stat.cmu.edu/datasets/wind.data.
The geographical locations  of the stations are shown in
\citetext{haslett:raftery:1989}; they are fairly uniformly distributed
over Ireland.  Each functional observation $X_n(s_k,t)$ consists of the
average of wind speed for day $t$, month $n$ ($N=216$), and at
location $s_k$. The left panel of Figure~\ref{f:figure1} shows the daily averages
for the $11$ stations for January $1961$.  The right panel shows a functional box plot of the same data \cite{sun:genton:2011}.

\begin{figure}[t]
\begin{center}
\includegraphics[scale=.5]{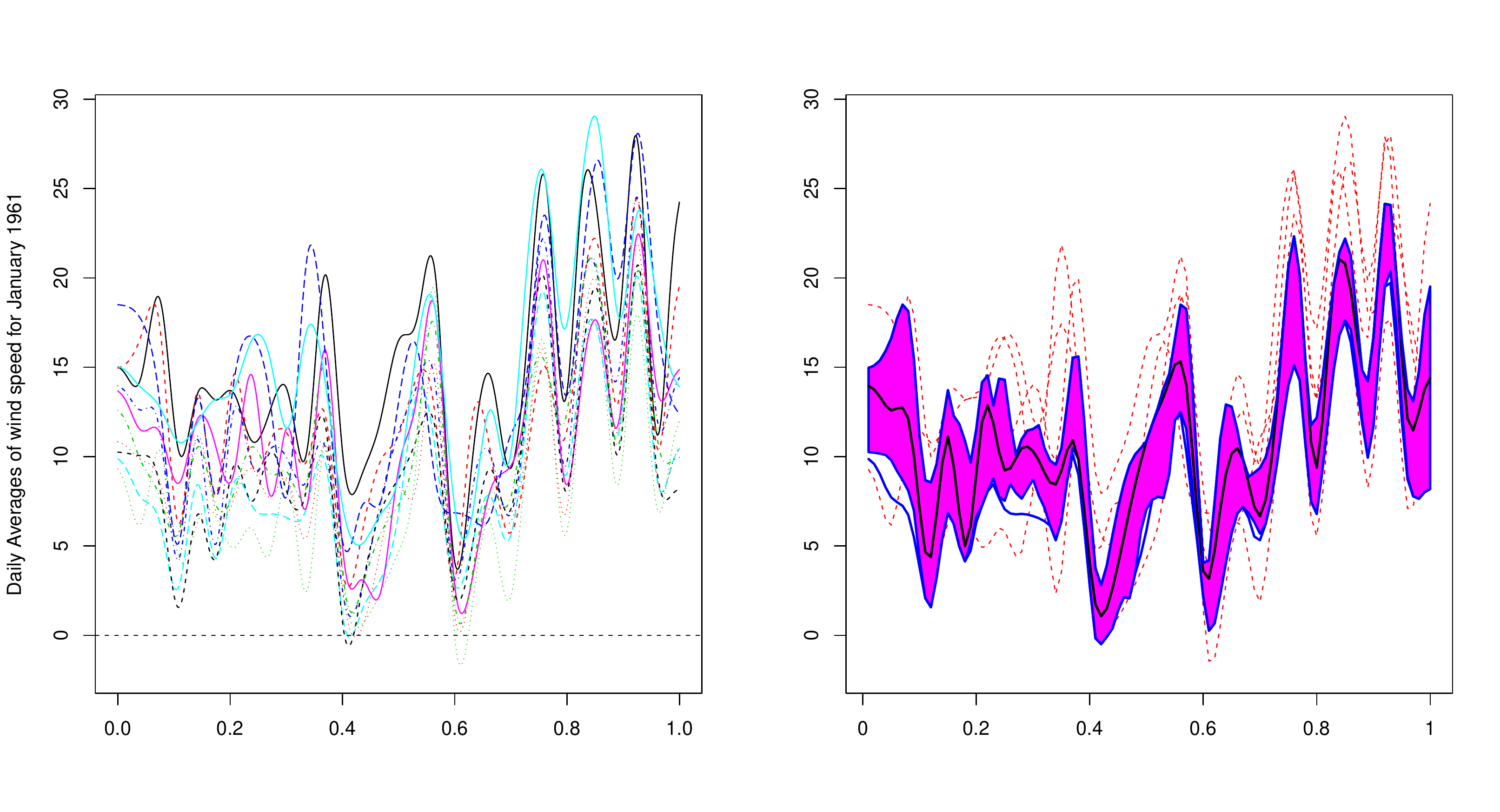}
\caption{Irish wind speed curves for January 1961.  Each is measured at a different location.  The left panel plots the functional observations while the right gives a functional boxplot. \label{f:figure1}}
\end{center}
\end{figure}

\citetext{gneiting:2002} estimated model \eqref{e:gneiting-cov}
on these data and obtained $\hat{\beta}=0.61$, which indicates a
nonseparable covariance structure. We  applied   our
tests to validate this conjecture. All four  tests produced
P--values smaller than 10E-4 for all $L, J \in \{ 2,3, 4 \}$.
The P--values of the $T_L$--test were all smaller than 10E-104 and
of the $T_W$--test smaller than 10E-21. The $T_F$--test had the
largest P--values (but still extremely small).

While the tests fully validate the conclusion of
\citetext{gneiting:2002}, we provide another illustration by applying
them  to residual curves obtained after removing the monthly mean
from each curve; we center all January months, February months, etc.,
separately. This simple transformation removes to a large
extent the annual seasonality, and it is interesting to
see if a nonseparable structure is still needed for the
data so transformed. The P--values for selected combinations
of $L$ and $J$ are shown in Table~\ref{tb:centered-P}.
We now see a much more interesting pattern.  Only when $L$ and $J$ are
larger than $3$, do we see a clear evidence for nonseparability.
A possible explanation  is that the covariance structure is  made
up of two components, one which is separable and one which is not.
The separable component makes up the majority of the variation in the
process which is why the separability is not seen for smaller values
of $L$ and $J$. The pattern of dependence of the P--values
on $L$ and $J$ is inconsistent with what we have seen in Scenario 3
above, but this may be due to the specific parameter values in
\eqref{e:gneiting-cov} we used in the simulations.

\begin{table} [h]
\begin{center}
\begin{tabular}{|c|c|c|c|c|}
\hline
& $T_{L-MC}$ & $T_L$ & $T_F$ & $T_W$\\
\hline
$L=J=2$ & 0.06 & 0.055 &  0.039 &   0.034\\
\hline
$L=J=3$ & 0.467 & 0.436 &   0.149 &    0.203\\
\hline
$L=J=4$ & $<$10E-6 &  8.079E-24  &   4.237E-09 &  2.146E-09\\
\hline
\end{tabular}
\end{center}
\caption{\label{tb:centered-P} P--values for the separability tests
applied to deseasonalized wind speed data.}
\end{table}

\clearpage

\appendix

\section{Derivation of the Q matrices} \label{ss:qmat}
This section introduces four  matrices that describe the covariance
structure of products of various vectorized matrices consisting of standard
normal variables. We refer to them collectively as ``Q matrices'', as
we use the symbol $\bQ$ with suitable subscripts and superscripts to
denote them. These matrices appear in the asymptotic distribution of
the vectorized matrices $\widehat\bU, \widehat\bV, \widehat\bSig$,
which, in turn, is used to prove Theorem~\ref{t:mult-asy}. In particular,
the asymptotic distribution of statistic $\widehat T_F$, which
we recommended in Section~\ref{s:sim},  is expressed in terms
of these Q matrices. Some of them are  defined though an algorithm.

\begin{theorem}
If $\bE$ is an $K \times I$ matrix of standard normals, then
\begin{align}
\Cov( \vc(\bE \bE^\top)) = 2I \bQ_K. \label{e:qs}
\end{align}
where
\[
 \bQ_K(i,j) =
\begin{cases}
1 & i = j = k + (k-1)K \mbox{ for } k=1,\dots,K \\
\frac{1}{2} & i = j \neq k + (k-1)K \mbox{ for } k=1,\dots,K \\
\frac{1}{2} & i \neq j = 1 + \left( \frac{(i-1)-((i-1)\mbox{ mod }K)}{K} \right) + ((i-1) \mbox{ mod } K) K \\
0 & \mbox{otherwise},
\end{cases}.
\]
\end{theorem}

\begin{proof}
Denote by $e_{kl}$ the independent standard normals, and set
\[
{\be}_k = [e_{k1}, e_{k2}, \ldots, e_{kI}]^\top,  \ \ \ 1 \le k \le K,
\]
so that
\[
\bE = \left[ \begin{matrix}
\be_1^\top \\
\vdots \\
\be_K^\top
\end{matrix} \right].
\]
Then for any $i,j$ in $\{1,\dots,K\}$ we have that the $(i,j)$ entry of
$\Cov(\vc( \bE \bE^\top))$ can be written as
\[
\Cov( \be_{k_1}^\top \be_{l_1}, \be_{k_2}^\top \be_{l_2}),
\]
where we have the relationships
\begin{align*}
& i = k_1 + (l_1 -1)K,  \qquad k_1 = ((i-1) \mbox{ mod } K) + 1, \qquad l_1 = 1+\frac{(i - 1) - (i-1) \mbox{ mod } K}{K} \\
& j = k_2 + (l_2 -1)K,  \qquad k_2 = ((j-1) \mbox{ mod } K) + 1, \qquad l_2 = 1+\frac{(j - 1) - (j-1) \mbox{ mod } K}{K}.
\end{align*}
For the diagonal terms, i.e. $i=j$, we have two settings $k_1 = k_2 =
l_1 = l_2$, in which case the covariance is $2 I$, or alternatively
$k_1 = k_2 \neq l_1 = l_2$ in which case the covariance is $I$.  Since
the former occurs in every $K^{th}$ term, we have established the
proper pattern for the diagonal.

We now need only establish the pattern for the off diagonal.  Every
term in the off diagonal can be expressed as $\Cov(\be_i^\top \be_j,
\be_k^\top \be_l)$, for some $i,j,k,l =1, \dots, K$.  Clearly, if any
one index is different from the other 3, then the covariance is 0.  We
can't have all four indices being equal as that would be a diagonal
element, and we can't have $i = k \neq j = k$ as that would also be a
diagonal element.  If $i=j$ and $k=l$ then two inner products are
independent, and thus the covariance is zero.  Therefore, the only
nonzero off diagonal entries occur when $i= l \neq j = k$, and the
covariance would be $I$.  To determine where in the $K^2 \times K^2$
matrix these occur, we use the change of base formulas.
\end{proof}

We illustrate the form of the matrices $\bQ_K$:
\[
\bQ_2 =
\left(
\begin{matrix}
1.0 & 0 & 0 & 0 \\
  0 & 0.5 & 0.5 & 0 \\
  0 & 0.5 & 0.5 & 0 \\
  0 & 0 & 0 & 1.0
\end{matrix}
\right),
\]
\[
\bQ_3 =
\left(
\begin{matrix}
1.0 & 0 & 0 & 0 & 0 & 0 & 0 & 0 & 0 \\
  0 & 0.5 & 0 & 0.5 & 0 & 0 & 0 & 0 & 0 \\
  0 & 0 & 0.5 & 0 & 0 & 0 & 0.5 & 0 & 0 \\
  0 & 0.5 & 0 & 0.5 & 0 & 0 & 0 & 0 & 0 \\
  0 & 0 & 0 & 0 & 1.0 & 0 & 0 & 0 & 0 \\
  0 & 0 & 0 & 0 & 0 & 0.5 & 0 & 0.5 & 0 \\
  0 & 0 & 0.5 & 0 & 0 & 0 & 0.5 & 0 & 0 \\
  0 & 0 & 0 & 0 & 0 & 0.5 & 0 & 0.5 & 0 \\
  0 & 0 & 0 & 0 & 0 & 0 & 0 & 0 & 1.0
  \end{matrix}
\right).
\]

\begin{theorem}
If $\bE$ is an $K \times I$ matrix of standard normals, then
\begin{align}
\Cov( \vc(\bE \bE^\top), \vc(\bE^\top \bE) ) = 2 \sqrt{IK} \bQ_{K,I},
\label{e:qst}
\end{align}
where $\bQ_{K,I}$ is an $K^2 \times I^2$ matrix given by
\[
\bQ_{K,I}(i,j) = \begin{cases}
(KI)^{-1/2} & i = k_1 + K(k_1-1), j = k_2 + I(k_2 -1)
\ k_1 = 1,\dots,K \ k_2 = 1,\dots, I\\
0 & \mbox{otherwise},
\end{cases}.
\]
\begin{proof}
  Here, each entry of the above covariance matrix is obtained by
  taking two rows of $\bE$ (possibly the same row) forming the inner
  product, taking two columns, taking their inner product, and then
  computing the covariance between the two.  Due to the symmetry of
  this calculation, there are only three possible resulting values:
  when the two rows are different, then the two columns are different,
  or when both the rows and columns are different.  When both rows and
  columns are different, we can (without loss of generality) take the
  first two rows and columns.  In that case, the covariance becomes
\[
\Cov \left( \sum_{k=1}^K e_{1,k} e_{2,k}, \sum_{i=1}^I e_{i,1} e_{i,2}  \right)
= \sum_{k=1}^K \sum_{i=1}^I \Cov \left(  e_{1,k} e_{2,k},  e_{i,1} e_{i,2}  \right).
\]
However, every summand above is zero when $i>2$ or $k>2$ since they will then involve independent variables.  Therefore, we can express the above as
\[
\Cov(e_{1,1}e_{2,1} + e_{1,2} e_{2,2}, e_{1,1}e_{1,2} + e_{2,1} e_{2,2}) = 0.
\]
Hence, any term with two different rows and columns is zero.  A similar result will hold when there are either two different rows or two different columns.  The only nonzero term will stem from taking the same row and same column, in which case the value becomes
\[
\Cov(e_{1,1}e_{1,1} + e_{1,2} e_{1,2}, e_{1,1}e_{1,1} + e_{2,1} e_{2,1}) = \Var(e_{1,1}^2) = 2.
\]
Therefore, every nonzero entry will be 2.  We now only need to determine which entries of the covariance matrix correpond to taking the same row and same column.  Considering the structure induced by vectorizing, the first row of the covariance matrix and every subsequent $K$ rows will correspond to matching the same row of $\bE$.  Similalry, the first and every subsequent $I$ column will correspond to matching the same colmun of $\bE$.  This corresponds to our definition and the result follows.

\end{proof}

\end{theorem}

Some examples $\bQ_{K,I}$ are
\[
\bQ_{2,2} =
\left(
\begin{matrix}
0.5 & 0 & 0 & 0.5 \\
  0 & 0 & 0 & 0 \\
  0 & 0 & 0 & 0 \\
  0.5 & 0 & 0 & 0.5
\end{matrix}
\right)
\]
and
\[
\bQ_{2,3} =
\left(
\begin{matrix}
6^{-1/2} & 0 & 0 & 0 & 6^{-1/2} & 0 & 0 & 0 & 6^{-1/2} \\
  0 & 0 & 0 & 0 & 0 & 0 & 0 & 0 & 0 \\
  0 & 0 & 0 & 0 & 0 & 0 & 0 & 0 & 0 \\
  6^{-1/2} & 0 & 0 & 0 & 6^{-1/2} & 0 & 0 & 0 & 6^{-1/2} \\
\end{matrix}
\right).
\]

\medskip

Before the next theorem, we define the matrix $\tilde \bQ_{R,K}$ via
a pseudo code. \\

\quad Begin Code

\qquad Set $\tilde \bQ_{R,K}$ to be an $R^2=K^2I^2$ by $K^2$ matrix of zeros

\quad \qquad For $i = 1,\dots, R^2$

\qquad \qquad  $l = 1 + \lfloor (i-1)/(K^2I) \rfloor$

\qquad \qquad $k = 1 + \lfloor(i-1-(l-1)K^2T)/(KI) \rfloor$
		
\qquad \qquad $p = 1 + \lfloor(i-1-(l-1)K^2I-(k-1)KI)/(K) \rfloor$
		
\qquad \qquad  $m = i-(l-1)K^2I-(k-1)KI-(p-1)K$

\qquad \qquad If ($p=  l $ and $m \neq k$) then $\tilde \bQ_{R,K}[i,m+(k-1)K] = 1/(2 \sqrt{I})$

\qquad \qquad \quad and  $\tilde \bQ_{R,K}[i,k+(m-1)K] = 1/(2 \sqrt{I})$

\qquad \qquad If ($p = l$ and $m = k$) then $\tilde \bQ_{R,K}[i,m+(m-1)K] = 1/ \sqrt{I}$

\quad \qquad End For Loop

\quad End Code

\begin{theorem}
If $\bE$ is an $K \times I$ matrix of standard normals and
$\bE_\diamond = \vc(\bE)$, then
\begin{align} \label{e:qtilde}
\Cov( \vc(\bE_\diamond \bE_\diamond^\top), \vc(\bE \bE^\top) ) = 2 \sqrt{T} \tilde \bQ_{R,K},
\end{align}
where $\tilde \bQ_{R,K}$ is the $R^2 \times K^2$ matrix defined
by the pseudo code above.
\begin{proof}
  Begin by considering the $(i,j)$ of the desired covariance matrix.
  There exists indices such that the $(i,j)$ entry is equal to
\[
\Cov(e_{m,p} e_{k,l}, e_r^\top e_s),
\]
where $m,k,r,s$ take values $1,\dots, K$ and $p,l$ take values $1,\dots,I$.   Moving from $(r,s)$ to $j$ we have that
\[
j = 1 + (r-1) + (s-1)K,
\]
and the reverse is obtained using
\begin{align*}
s & = 1 + \lfloor (j-1)/K \rfloor \\
r & = j-(s-1)K.
\end{align*}
Moving from $(m,p,k,l)$ to $i$ we have that
\[
i = 1+(m-1) + (p-1)K + (k-1)KI + (l-1)K^2I.
\]
We can move back to $(m,p,k,l)$ from $i$ using
\begin{align*}
l & = 1 + \lfloor (i-1)/(K^2I) \rfloor \\
k & = 1 + \lfloor(i-1-(l-1)K^2I)/(KI) \rfloor \\	
p & = 1 + \lfloor(i-1-(l-1)K^2I-(k-1)KI)/(K) \rfloor \\
m &= i-(l-1)K^2I-(k-1)KI-(p-1)K.
\end{align*}

We can see that the covariance will be zero if any one of $m,k,r$ or $s$ is distinct.  Thus, the only nonzero entries will correspond to either $m = r = k =s$, $m = r \neq k = s$, or $m = s \neq k = r$ (when $m=k \neq r = s$ we get zero).  When all four are equal we get that
\[
\Cov(e_{m,p} e_{m,l}, e_m^\top e_m) = \Cov(e_{m,p} e_{m,l}, e_{m,p}^2 +  e_{m,l}^2 1_{p \neq l}),
\]
which will be zero unless $p=l$, in which case it equals
\[
\Cov(e_{m,p}^2, e_{m,p}^2 ) = (\E[e_{m,p}^4] - \E[e_{m,p}^2]^2) = 2.
\]
We have therefore established the first \textit{if}-statement in the pseudo code.

Turning to the next case, when $ m =r \neq k =s$, we have that
\[
\Cov(e_{m,p} e_{k,l}, e_m^\top e_k)= \Cov(e_{m,p} e_{k,l}, e_{m,p} e_{k,p} + e_{m,l} e_{k,l}1_{l \neq p}),
\]
which is again only nonzero when $p = l$, in which case it will be
\[
 \Cov(e_{m,p} e_{k,p}, e_{m,p} e_{k,p}) = \E[e_{m,p}^2]\E[ e_{k,p}^2] = 1.
\]
An identical result will hold for when $m =s \neq k =r$, which gives both the second and third \textit{if}-statements in the pseudo code, and the proof is established.
\end{proof}
\end{theorem}

One example of $\tilde \bQ_{R,K}$ is
\[
\tilde \bQ_{4,2}
 = \left(
\begin{matrix}
 2^{-1/2} & 0 & 0 & 0 \\
  0 & 8^{-1/2} & 8^{-1/2} & 0 \\
  0 & 0 & 0 & 0 \\
  0 & 0 & 0 & 0 \\
  0 & 8^{-1/2} & 8^{-1/2} & 0 \\
  0 & 0 & 0 & 2^{-1/2} \\
  0 & 0 & 0 & 0 \\
  0 & 0 & 0 & 0 \\
  0 & 0 & 0 & 0 \\
  0 & 0 & 0 & 0 \\
  2^{-1/2} & 0 & 0 & 0 \\
  0 & 8^{-1/2} & 8^{-1/2} & 0 \\
  0 & 0 & 0 & 0 \\
  0 & 0 & 0 & 0 \\
  0 & 8^{-1/2} & 8^{-1/2} & 0 \\
  0 & 0 & 0 & 2^{-1/2}
  \end{matrix}
\right).
\]

For the last Q matrix, we also use pseudo code which is only slightly
different  from the code defining $\tilde \bQ_{R, K}$.

\quad Begin Code

\qquad Set $\breve \bQ_{R,I}$ to be an $R^2=K^2I^2$ by $I^2$ matrix of zeros

 \qquad For $i = 1,\dots, R^2$

\quad \qquad   $l = 1 + \lfloor (i-1)/(K^2I) \rfloor$

\quad \qquad $k = 1 + \lfloor(i-1-(l-1)K^2I)/(KI) \rfloor$
		
\quad \qquad $p = 1 + \lfloor(i-1-(l-1)K^2I-(k-1)KI)/(K) \rfloor$
		
\quad \qquad  $m = i-(l-1)K^2I-(k-1)KI-(p-1)K$

\quad \qquad If ($m =  k $ and $p \neq l$) then $\breve \bQ_{R,I}[i,p+(l-1)I] = 1/(2 \sqrt{K})$

\quad \qquad \quad and  $\breve \bQ_{R,I}[i,l+(p-1)I] = 1/(2 \sqrt{K})$

\quad \qquad If ($m = k$ and $p = l$) then $\breve \bQ_{R,I}[i,l+(l-1)I] = 1/ \sqrt{K}$

\qquad End For Loop

\quad End Code

\begin{theorem}
If $\bE$ is an $K \times I$ matrix of standard normals and
$\bE_\diamond = \vc(\bE)$,  then
\begin{align} \label{e:qbreve}
\Cov( \vc(\bE_\diamond \bE_\diamond^\top), \vc(\bE^\top \bE) ) = 2 \sqrt{K} \breve\bQ_{R,I},
\end{align}
where $\breve \bQ_{R,I}$ is the  $R^2 \times I^2$ matrix
defined by the pseudo code above.
\end{theorem}
\begin{proof}
Begin by considering the $(i,j)$ of the desired covariance matrix.  There exists indices such that the $(i,j)$ entry is equal to
\[
\Cov(e_{m,p} e_{k,l}, e_{(r)}^\top e_{(s)}),
\]
where $m,k$ take values $1,\dots, K$ and $p,l,r,s$ take values $1,\dots,I$.   Moving from $(r,s)$ to $j$ we have that
\[
j = 1 + (r-1) + (s-1)I,
\]
and the reverse is obtained using
\begin{align*}
s & = 1 + \lfloor (j-1)/I \rfloor \\
r & = j-(s-1)I.
\end{align*}
Moving from $(m,p,k,l)$ to $i$ we have that
\[
i = 1+(m-1) + (p-1)K + (k-1)KI + (l-1)K^2I.
\]
We can move back to $(m,p,k,l)$ from $i$ using
\begin{align*}
l & = 1 + \lfloor (i-1)/(K^2I) \rfloor \\
k & = 1 + \lfloor(i-1-(l-1)K^2I)/(KI) \rfloor \\	
p & = 1 + \lfloor(i-1-(l-1)K^2I-(k-1)KI)/(K) \rfloor \\
m &= i-(l-1)I^2K-(k-1)KI-(p-1)K.
\end{align*}
We can see that the covariance will be zero if any one of $p,l,r$ or $s$ is distinct.  Thus, the only nonzero entries will correspond to either $p = r = l =s$, $p = r \neq l = s$, or $p = s \neq l = r$ (when $p=l \neq r = s$ we get zero).  When all four are equal we get that
\[
\Cov(e_{m,p} e_{k,p}, e_{(p)}^\top e_{(p)}) = \Cov(e_{m,p} e_{k,p}, e_{m,p}^2 +  e_{k,p}^2 1_{m \neq k}),
\]
which will be zero unless $m=k$, in which case it equals it will be equal to 2.
We have therefore established the first \textit{if}-statement in the
 pseudo code.

Turning to the next case, when $ p =r \neq l =s$, we have that
\[
\Cov(e_{m,p} e_{k,l}, e_{(p)}^\top e_{(l)})= \Cov(e_{m,p} e_{k,l}, e_{m,p} e_{m,l} + e_{k,p} e_{k,l}1_{m \neq k}),
\]
which is again only nonzero when $m = k$, in which case it will be
\[
 \Cov(e_{m,p} e_{m,l}, e_{m,p} e_{m,l}) = \E[e_{m,p}^2]\E[ e_{m,l}^2] = 1.
\]
An identical result will hold for when $p =s \neq l =r$, which gives
both the second and third \textit{if}-statements in the pseudo code, and
the proof is established.
\end{proof}

\section{Proof of Theorem~\ref{t:mult-asy}}
\label{ss:mult}
We begin by establishing in Theorem~\ref{t:joint} the joint null limit
distribution of the vectors $\vc(\hat \bU - \bU), \vc(\hat \bV -
\bV)$, and $\vc(\hat \bSig - \bSig)$. We first define several matrices
that appear in this distribution.  Recall  the Q
matrices derived  in Section~\ref{ss:qmat}: the matrix $\bQ_K$ is
defined in \eqref{e:qs}, $\bQ_{K,I}$ in \eqref{e:qst}, $\tilde
\bQ_{R,K}$ in \eqref{e:qtilde}, and $\breve \bQ_{R,I}$ in
\eqref{e:qbreve}. Denote by  $(\cdot)^+$  a generalized inverse.
We  define the following generalized information
matrices:
\aln{
\cI_{\bU,\bV} =
& \frac{1}{2}
\left(
\begin{matrix}
\bU^{-1/2} \otimes \bU^{-1/2} & 0 \\
0& \bV^{-1/2} \otimes \bV^{-1/2}
\end{matrix}
\right) \\
& \times \left(
\begin{matrix}
I \bQ_K & \sqrt{IK} \bQ_{K,I} \\
\sqrt{IK} \bQ_{I,K}& K \bQ_I
\end{matrix}
\right)
\left(
\begin{matrix}
\bU^{-1/2} \otimes \bU^{-1/2} & 0 \\
0& \bV^{-1/2} \otimes \bV^{-1/2}
\end{matrix}
\right), \\
\cI_{\bU,\bV}^c
= &  (\bD (\bD^\top \cI_{\bU, \bV} \bD)^{+} \bD^\top)^+,
}
where $\bD$ is an $(K^2+I^2) \times (K^2+I^2 - 1)$ matrix whose
columns are orthonormal and are perpendicular to $\vc(\bI_{K^2+I^2})$,
and
\aln{
\cI_{\bSig} = \frac{1}{2} (\bSig^{-1/2} \otimes \bSig^{-1/2}) \bQ_{R} (\bSig^{-1/2} \otimes \bSig^{-1/2}),
}
where $R = KI$.

\begin{theorem}\label{t:joint}  Suppose Assumption~\ref{a:mult}
and decomposition \eqref{e:bSigDec} hold.
Assume further that $\tr(\bU) = K$. Then
\begin{align*}
\sqrt{N} \left( \begin{matrix}
\vc(\hat \bU - \bU) \\
\vc(\hat \bV - \bV) \\
\vc(\hat \bSig - \bSig)
\end{matrix}
\right)
\overset{\cD}{\to} N( {\bf 0}, \bGg).
\end{align*}
The asymptotic covariance matrix $\bGg$ is defined as follows.  The
asymptotic covariance of $(\vc(\hat \bU - \bU), \vc(\hat \bV - \bV))$
is given by $ (\cI_{\bU,\bV}^c)^+$, of $ \vc(\hat \bSig - \bSig)$ is
given by $\cI_{\bSig}^+$, and the cross covariance matrix between the
two is
\[
 \frac{1}{2}  (\cI^{c}_{\bU,\bV})^{+}
\left(
\begin{matrix}
\bU^{-1/2} \otimes \bU^{-1/2} & 0 \\
0& \bV^{-1/2} \otimes \bV^{-1/2}
\end{matrix}
\right)
\left( \begin{matrix}
\sqrt{I} \tilde \bQ_{R,K}^\top \\
\sqrt{K} \breve \bQ_{R,I}^\top
\end{matrix} \right)
( \bSig^{-1/2} \otimes \bSig^{-1/2})
  \cI^{+}_{\bSig}.
\]
\end{theorem}
\begin{proof}
From standard theory for MLEs, \citetext{ferguson:1996} Chapter 18, we
can use the partial derivatives of the log likelihood function
(score equations) to find the Fisher information as well as asymptotic
expressions for the MLEs.  One can show that the cross terms of the
Fisher information involving $\bM$ and $\bV, \bU$ and $\bSig$ are all zero,
meaning that the estimate of the $\bM$ is asymptotically independent
of $\widehat \bU, \widehat\bV$ and $\widehat \bSig$. We therefore
 treat in the following $\bM$ as known.

  We start by working with $\bU$ and $\bV$.  Applying the constrained
  likelihood methods described in \citetext{moore:etal:2008},
  asymptotically, $\widehat \bU$ and $\widehat \bV$ are jointly normally
  distributed with means $\bU$ and $\bV$ and covariance given by the
  generalized inverse of the constrained Fisher information matrix.
  Starting with $\bU$ we have that unconstrained score equation is
  given by \aln{
    & \frac{1}{\sqrt{N}}\frac{\partial l(\bM, \bU, \bV)}{\partial \bU} \\
    & =
    \frac{1}{2 \sqrt{N}} \sum_{n=1}^N [\bU^{-1}   (\bX_n - \bM) \bV^{-1} (\bX_n - \bM)^\top \bU^{-1} - T \bU^{-1}]  \\
    & =
    \frac{1}{2 \sqrt{N}} \sum_{n=1}^N [\bU^{-1/2}   \bE_n \bE_n^\top \bU^{-1/2} - T \bU^{-1}]  \\
    & = \frac{1}{2 \sqrt{N}} \sum_{n=1}^N [\bU^{-1/2} [\bE_n
    \bE_n^\top - I \bI_{K \times K} ]\bU^{-1/2} ].  } To get a handle
  on the unconstrained Fisher information matrix (and therefore the
  covariance matrix), it will be easier to work with the vectorized
  version
\[
\frac{1}{2 \sqrt{N}} (\bU^{-1/2} \otimes \bU^{-1/2} )\sum_{n=1}^N \vc[ \bE_n \bE_n^\top - I \bI_{K \times K} ].
\]
Notice that we will have a complete handle on the above if we can
understand the form for the covariance of $\vc[ \bE_n \bE_n^\top - I
\bI_{K \times K} ]$.  However, this is a term that in no way depends
on the underlying parameters as it is composed entirely of iid
standard normals.  We label $ \bQ_K = (2 I)^{-1} \Cov(\vc[ \bE_n
\bE_n^\top - I \bI_{K \times K} ])$ and its explicit form is given in
\eqref{e:qs}.

The part of the Fisher information matrix for $\vc(\bU)$ is given by
\[
\frac{I}{2}
(\bU^{-1/2} \otimes \bU^{-1/2} )
\bQ_K
(\bU^{-1/2} \otimes \bU^{-1/2} ).
\]
Identical arguments give that the part of the Fisher information matrix for $\vc(\bV)$ is given by
\[
\frac{K}{2}
(\bV^{-1/2} \otimes \bV^{-1/2} )
\bQ_I
(\bV^{-1/2} \otimes \bV^{-1/2} ).
\]
The joint unconstrained Fisher information matrix for $\widehat \bU$ and $\widehat \bV$ is given by
\aln{
\cI_{\bU,\bV} =
\frac{1}{2}
&\left(
\begin{matrix}
\bU^{-1/2} \otimes \bU^{-1/2} & 0 \\
0& \bV^{-1/2} \otimes \bV^{-1/2}
\end{matrix}
\right) \\
& \times \left(
\begin{matrix}
I \bQ_K & \sqrt{IK} \bQ_{I,K} \\
\sqrt{IK} \bQ_{I,K}& K \bQ_I
\end{matrix}
\right)
\left(
\begin{matrix}
\bU^{-1/2} \otimes \bU^{-1/2} & 0 \\
0& \bV^{-1/2} \otimes \bV^{-1/2}
\end{matrix}
\right),
}
where $\bQ_{K,I}$ is defined in \eqref{e:qst}.  The constrained version is then given by
\[
\cI_{\bU,\bV}^c = (\bD (\bD^\top \cI_{\bU, \bV}\bD) \bD^+)^+.
\]
Recall that $\bD$ is an $(K^2+I^2) \times (K^2+I^2 - 1)$ matrix whose
columns are orthonormal and are perpendicular to $\vc(\bI_{K^2+I^2})$.
The form for $\bD$ come from the gradient of the constraint $\tr(\bU) = K$.

The last piece we need is the joint behavior of $\widehat \bU$ (or
$\widehat \bV$) and the estimator $\widehat \bSig$.  The score
equation for $\bSig$ can be expressed as \aln{ \frac{\partial l(\mu,
    \bSig)}{\partial \bSig} & = - \frac{N}{2} \bSig^{-1}
  + \frac{1}{2} \bSig^{-1} \left( \sum_{n=1}^N  (Y_n - \mu)  (Y_n-\mu) \right) \bSig^{-1} \\
  & = - \frac{N}{2} \bSig^{-1}
  + \frac{1}{2} \bSig^{-1/2} \left( \sum_{n=1}^N  E_n  E_n^\top \right) \bSig^{-1/2} \\
  & = \frac{1}{2} \bSig^{-1/2} \left( \sum_{n=1}^N E_n E_n^\top -
    \bI_{KI \times KI}\right) \bSig^{-1/2}.  } Using the same
arguments as before, we get that Fisher information matrix for
$\vc(\bSig)$ is \aln{ \cI_\bSig = \frac{1}{2} (\bSig^{-1/2} \otimes
  \bSig^{-1/2}) \bQ_{R} (\bSig^{-1/2} \otimes \bSig^{-1/2}).  } For
the joint behavior, we use the following asymptotic expression for the
MLEs, \citetext{ferguson:1996} Chapter 18, \[ \sqrt{N} \left(
\begin{matrix}
 \vc(\widehat \bU - \bU) \\
\vc(\widehat \bV - \bV) \\
\end{matrix}
\right)
= \frac{1}{\sqrt{N}} (\cI^{c}_{\bU,\bV})^+
\left(
\begin{matrix}
\vc \left(    \frac{\partial l(\bM, \bU, \bV)}{\partial \bU} \right) \\
\vc \left(   \frac{\partial l(\bM, \bU, \bV)}{\partial \bV} \right) \\
\end{matrix}
\right) + o_P(1),
\]
and
\[
\sqrt{N}\vc(\widehat \bSig - \bSig) = \frac{1}{\sqrt{N}} \cI^{+}_{\bSig} \frac{\partial l(\mu, \bSig)}{\partial \bSig} + o_P(1).
\]
For the covariance between $\widehat \bSig$ and $\widehat \bU$ (or $\widehat \bV$) we obtain two more matrices, called $\tilde \bQ_{R,K}$ and $\breve \bQ_{R,I}$ which satisfy
\begin{align*}
\Cov \left[\vc(\bE_{\diamond n}\bE_{\diamond n}^\top
- \bI_{R \times R}), \vc(\bE_n\bE_n^\top - I \bI_{K \times K})\right]
= 2 \sqrt{I} \tilde \bQ_{R,K}, \\
\Cov \left[\vc(\bE_{\diamond n} \bE_{\diamond n}^\top
- \bI_{R \times R}), \vc(\bE_n^\top \bE_n - K\bI_{I \times I})\right]
= 2 \sqrt{K} \breve \bQ_{R,I}.
\end{align*}
Recall that the diamond subscript indicates vectorization and the
definitions of $\tilde \bQ_{R,K}$ and $\breve \bQ_{R,I}$ can be found
in \eqref{e:qtilde} and \eqref{e:qbreve}, respectively.  The
cross--covariance matrix for $\sqrt{N}\vc(\widehat \bSig - \bSig )$
and $\sqrt{N}\vc(\widehat \bU - \bU)$ is then given by
\aln{ &
  (\cI^{c}_{\bU,\bV})^{+} \Cov \left( \left(
\begin{matrix}
\vc \left(   N^{-1/2} \frac{\partial l(M, U, V)}{\partial U} \right) \\
\vc \left(  N^{-1/2}  \frac{\partial l(M, U, V)}{\partial V} \right) \\
\end{matrix}
\right) ,
\vc\left(N^{-1/2} \frac{\partial l(M, \Sigma)}{\partial \Sigma}\right)
 \right)
  \cI^{+}_{\bSig} \\
  & = \frac{1}{2}  (\cI^{c}_{\bU,\bV})^{+}
\left(
\begin{matrix}
\bU^{-1/2} \otimes \bU^{-1/2} & 0 \\
0& \bV^{-1/2} \otimes \bV^{-1/2}
\end{matrix}
\right)
\left( \begin{matrix}
\sqrt{I} \tilde \bQ_{R,K}^\top \\
\sqrt{K} \breve \bQ_{R,I}^\top
\end{matrix} \right)
( \bSig^{-1/2} \otimes \bSig^{-1/2})
  \cI^{+}_{\bSig} \\
}
\end{proof}

\medskip

\noindent{\it Proof of Theorem~\ref{t:mult-asy}:}
Since we have the joint asymptotic distribution for $\widehat \bU,
\widehat \bV,$ and $\widehat \bSig$, we can use the delta method to
find the asymptotic distributions of desired test statistics, and in
particular, we can find the form of $\bW$, the asymptotic covariance
matrix of $ \vc( \widehat \bV \otimes \widehat \bU) - \vc( \widehat
\bSig).  $ To apply the delta method, we need the partial derivatives.
Taking the derivative with respect to $V_{i,j}$ yields
\[
\vc(\bone_{i,j} \otimes \bU)
\]
and with respect to $U_{k,l}$
\[
\vc(\bV \otimes \bone_{k,l}).
\]
So the matrix of partials with respect to $\vc(\bV)$ is
\[
\bG_\bV = \left (\begin{matrix}
 \vc(\bone_{1,1} \otimes \bU)^\top  \\
   \vc(\bone_{2,1} \otimes \bU)^\top \\
    \vdots  \\
     \vc(\bone_{I,I} \otimes \bU)^\top
\end{matrix} \right),
\]
with respect to $\vc(\bU)$ is
\[
\bG_\bU = \left (\begin{matrix}
\vc(\bV \otimes \bone_{1,1})^\top \\
 \vc(\bV \otimes \bone_{2,1})^\top \\
  \vdots &  \\
   \vc(\bV \otimes \bone_{K,K})^\top
\end{matrix} \right),
\]
and with respect to $\vc(\bSig)$ is just $(-1)$ times the $KI \times KI$ identity matrix.  We therefore have that
\[
\vc( \widehat \bV \otimes  \widehat \bU) - \vc( \widehat \bSig) \approx
 \left (\begin{matrix}
 \bG_{\widehat \bU} \\
 \bG_{\widehat \bV} \\
- \bI_{ST \times ST}
 \end{matrix} \right)^\top
 \left(\begin{matrix}
 \vc(\widehat \bU) \\
 \vc(\widehat \bV) \\
\vc(\widehat \bSig)
 \end{matrix} \right).
\]
This implies that
\begin{equation}  \label{e:wmatrix}
 \bW  =
\left (\begin{matrix}
 \bG_{ \bU} \\
 \bG_{ \bV} \\
- \bI_{KI \times KI}
 \end{matrix} \right)^\top
 \bGa
 \left (\begin{matrix}
 \bG_{ \bU} \\
 \bG_{\bV} \\
- \bI_{KI \times KI}
 \end{matrix} \right).
\end{equation}

The degrees of freedom are obtained by noticing that under the
alternative $\bSig$ has $KI(KI+1)/2$ free parameters, while under the
null there $ K(K +1)/2+ I(I +1)/2-1$, where the last $-1$ is
included because we have one constraint ($\tr(\bU) = K$).

\small

\bibliographystyle{oxford3}
\bibliography{sep}

\end{document}